\documentclass[12pt]{article}

\usepackage{fullpage}
\usepackage{amsfonts}
\usepackage{amsmath}
\usepackage{amsthm}
\usepackage{hyperref}

\newtheorem{lemma}{Lemma}
\newtheorem{theorem}{Theorem}

\newcommand{\ab}{\mathbf{a}}
\newcommand{\bb}{\mathbf{b}}

\newcommand{\xb}{\mathbf{x}}
\newcommand{\yb}{\mathbf{y}}
\newcommand{\zb}{\mathbf{z}}

\newcommand{\zerob}{\mathbf{0}}

\newcommand{\Ab}{\mathbf{A}}

\newcommand{\transpose}{{\mathrm{T}}}

\newcommand{\Fbb}{\mathbb{F}}
\newcommand{\Qbb}{\mathbb{Q}}

\begin{document}

\title{{\bf A Note on Semi-Algebraic Proofs \\
and Gaussian Elimination over Prime Fields}}
\author{Albert Atserias \\ Universitat Polit\`ecnica de Catalunya}

\maketitle

\begin{abstract}
  In this note we show that unsatisfiable systems of linear equations
  with a constant number of variables per equation over prime finite
  fields have polynomial-size constant-degree semi-algebraic proofs of
  unsatisfiability. These are proofs that manipulate polynomial
  inequalities over the reals with variables ranging in $\{0,1\}$.
  This upper bound is to be put in contrast with the known fact that,
  for certain explicit systems of linear equations over the
  two-element field, such refutations require linear degree and
  exponential size if they are restricted to so-called static
  semi-algebraic proofs, and even tree-like semi-algebraic and
  sums-of-squares proofs. Our upper bound is a more or less direct
  translation of an argument due to Grigoriev, Hirsch and Pasechnik
  (Moscow Mathematical Journal, 2002) who did it for a family of
  linear systems of interest in propositional proof complexity. We
  point out that their method is more general and can be thought of as
  simulating Gaussian elimination.
\end{abstract}

\section{Semi-algebraic proofs}

% For a natural number $n$, write $[n] := \{1,\ldots,n\}$.  If $\xb$ is
% a vector with $n$ components and $i \in [n]$, we write $x_i$ for its
% $i$-th component.  In the following, think of $\xb$ as a tuple of
% formal variables.  A \emph{monomial} is a product of powers of
% variables: $\prod_{i \in I} x_i^{\alpha_i}$, where $I \subseteq [n]$
% and each $\alpha_i$ is a natural number. The (total) degree of this
% monomial is $\sum_{i\in I} \alpha_i$. The (individual) degree of $x_i$
% in the monomial is $\alpha_i$. A \emph{polynomial} is a linear
% combination of monomials. The (total) degree of the polynomial is the
% maximum of the total degrees of its monomials, and the (individual)
% degree of a variable is defined similarly. A polynomial is linear if
% its degree is at most one, and linear in a variable if the degree of
% the variable is at most one. In the following we write arithmetic
% expressions, such as $x\cdot (1-x)$, that we identify with their
% underlying polynomial, in this case $x-x^2$.

The proof system we consider is inspired by the Sherali-Adams and
Lov\'asz-Schrijver lift-and-project methods for combinatorial
optimization \cite{SheraliAdams1990, LovaszSchrijver1991}, seen as
proof systems for deriving polynomial inequalities (see also
\cite{Pudlak1998, GrigorievHirschPasechnik2002,
  PitassiSegerlind2012}). In addition to the axioms
$$
x_i \geq 0 \;\;\;\;\;\;\;\;\;\;\; 
1 - x_i \geq 0 \;\;\;\;\;\;\;\;\;\;\;  
x_i^2 - x_i \geq 0 \;\;\;\;\;\;\;\;\;\;\;  
x_i - x_i^2 \geq 0 
$$ 
for formal variables $x_1,\ldots,x_n$, it has the following
inference rules:
$$
\frac{P(\xb) \geq 0 \;\;\;\;\;\;\;\;\;\; Q(\xb) \geq 0}{a\cdot P(\xb) +
  b\cdot Q(\xb) \geq 0} \;\;\;\;\;\;\;\;\; \frac{P(\xb) \geq 0}{P(\xb)
  \cdot x_i \geq 0} \;\;\;\;\;\;\;\;\; \frac{P(\xb) \geq 0}{P(\xb)
  \cdot (1-x_i) \geq 0}
$$ where $P(\xb)$ and $Q(\xb)$ are polynomials with rational
coefficients and variables within $\xb = (x_1,\ldots,x_n)$, and $a$
and $b$ are non-negative rational numbers. The first rule is called
\emph{positive linear combination} and the second and third rules are
called \emph{multiplication} or \emph{lifting rules}. It follows from
\cite{LovaszSchrijver1991} that if a system of linear inequalities
over the reals in $n$ variables does not have any solution in
$\{0,1\}^n$, then the trivial contradiction $-1 \geq 0$ can be derived
from the given inequalities, even if all polynomials are restricted to
total degree two. In general, the length of such a proof could be
exponential in a polynomial in $n$, but the shortest such proof is
never worse than that. Here, by length we mean the number of derived
inequalities. This and other complexity measures for proofs are
defined next.

% An extension of the basic semi-algebraic proof system could allow
% axioms $P(\xb)^2 \geq 0$ for any polynomial $P(\xb)$. This could make
% certain proofs shorter. The \emph{static} and was defined in \cite{GrigorievVo% robyov} and
% considered in \cite{Pudlak,GHP02,SegerlindPitassi,BBHKSZ} a
% similar proof system called sums-of-squares (SOS) is considered. In SoS
% a refutation of $P_1(\xb) = 0, \ldots,P_m(\xb) = 0$ is a collection of
% polynomials $Q_1(\xb),\ldots,Q_m(\xb)$, $R_1(\xb),\ldots,R_r(\xb)$
% such that $\sum_{i=1}^m Q_i(\xb) \cdot P_i(\xb) + \sum_{i=1}^r
% R_i(\xb)^2 = -1$. Let us define the size of this refutation as the sum
% of the sizes of the polynomials $Q_i(\xb)$, $Q_i(\xb) \cdot P_i(\xb)$
% for $i \in [m]$, and $R_j(\xb)$, $R_j(\xb)^2$ for $j \in [r]$. Its
% degree is the maximum of the degrees of these polynomials. It is not
% hard to see that if all the polynomials $x_i^2 - x_i$ are among
% $P_1(\xb),\ldots,P_m(\xb)$, then whenever the system $P_1(\xb) =
% 0,\ldots, P_m(\xb) = 0$ has an SoS refutation, it also has a
% semi-algebraic refutation with the additional axioms $R_i(\xb)^2 \geq
% 0$. Moreover the sizes of the refutation are polynomially related, the
% degree is the same, and the refutation is actually tree-like. The
% converse is also true except that the size and degree could be much
% bigger in SoS, basically because in semi-algebraic proofs we are
% allowed to reuse derived polynomials, but not in SoS.

The \emph{degree} of a proof is the maximum of the total degrees of
the polynomials that appear in it. The \emph{length} of a proof is the
number of inferences. The \emph{size} of a proof is the sum of the
sizes of the polynomials that appear in it, where the size of a
polynomial is the sum of the degrees of its monomials.  A proof is
\emph{tree-like} if every derived inequality is used at most once as
the hypothesis of another rule, i.e.~the shape of the proof is a tree,
with the hypotheses and the axioms at the leaves and the conclusion at
the root. A \emph{refutation} is a proof of $-1 \geq 0$. When we write
an inequality $P(\xb) \geq Q(\xb)$ what we really mean is $P(\xb) -
Q(\xb) \geq 0$. Similarly, when we write an equation $P(\xb) = Q(\xb)$
what we really mean is the set of the two inequalities $P(\xb) -
Q(\xb) \geq 0$ and $Q(\xb) - P(\xb) \geq 0$.

\section{Some facts about semi-algebraic proofs}

For every linear form $L(\xb) = \sum_{i=1}^n a_i x_i$ with rational
coefficients and every integer $c$, let $D_c(L(\xb))$ be the quadratic
polynomial $(L(\xb) - c)\cdot (L(\xb) - c + 1)$. In words, the
inequality $D_c(L(\xb)) \geq 0$ states that $L(\xb)$ does not fall in
the open interval $(c-1,c)$. Such statements have short proofs of low
degree:
 
\begin{lemma}[Grigoriev, Hirsch, and Pasechnik
  \cite{GrigorievHirschPasechnik2002}] \label{lem:GHP} For every
  integer $c$ and for every linear form $L(\xb) = \sum_{i=1}^n a_i
  x_i$ with integer coefficients $a_1,\ldots,a_n$, the inequality
  $D_c(L(\xb)) \geq 0$ has a tree-like proof of length polynomial in
  $\max\{|a_i| : i=1,\ldots,n\}$ and $n$, and degree at most~$3$.
\end{lemma}

The next lemma states that polynomial equalities can be freely
substituted. A similar statement appears in
\cite{GrigorievHirschPasechnik2002}[Lemma 5.2]; our statement is
slightly stronger.

\begin{lemma} \label{lem:subs} Let $P(\xb)$, $Q(\xb)$, and $R(\xb,y)$
  be polynomials with variables as indicated, and let $d$ be the
  degree of $y$ in $R(\xb,y)$.  The equation $R(\xb,P(\xb)) =
  R(\xb,Q(\xb))$ has a proof from $P(\xb) = Q(\xb)$ of length bounded
  by a degree-$d$ polynomial in the sizes of $P(\xb)$, $Q(\xb)$ and
  $R(\xb,y)$, and degree at most linear in the degree of $R(\xb,y)$
  and $d$ times the degrees of $P(\xb)$ and $Q(\xb)$.
\end{lemma}

\begin{proof}
  It suffices to prove the statement when $R(\xb,y)$ is linear in $y$;
  the general statement in which $y$ has degree $d \geq 1$ in
  $R(\xb,y)$ follows from iterating the lemma on the polynomial
  $R'(\xb,y_1,\ldots,y_d)$ obtained from $R(\xb,y)$ by replacing each
  $y^s$ by $\prod_{i=1}^s y_i$. Write every monomial in $R(\xb,y)$ in
  the form $y \cdot M(\xb)$, where $M(\xb)$ is a monomial without
  $y$. The equality $P(\xb) \cdot M(\xb) = Q(\xb) \cdot M(\xb)$
  follows at once from $P(\xb) = Q(\xb)$ by the multiplication rule.
  Adding up over all monomials of $R(\xb,y)$ we get the result.  To
  see the bound on the size note that, in case $R(\xb,y)$ is linear in
  $y$, the size of $R(\xb,P(\xb))$ is bounded by the product of the
  sizes of $R(\xb,y)$ and $P(\xb)$, and similarly for $R(\xb,Q(\xb))$.
  Iterating $d$ times to handle the general $d \geq 1$ case we get the
  degree-$d$ polynomial bound on the length of the proof.
\end{proof}

\section{Two-element field}

We identify the elements of the two-element field $\Fbb_2$ with
$\{0,1\}$.  Let $\xb = (x_1,\ldots,x_n)$ be formal variables ranging
over $\Fbb_2$ or $\Qbb$, depending on the context.  For every linear
equation of the form $\ab^\transpose \xb = b$, where $\ab \in
\Fbb_2^n$ and $b \in \Fbb_2$, let $\mathcal{S}(\ab,b)$ be the system of
linear inequalities
\begin{align*}
  & \sum_{i \in T} (1-x_i) + \sum_{i\in I\setminus T} x_i \geq 1 &
  \text{ for all } T \subseteq I \text{ such that } |T| \equiv
  1-b\!\!\mod 2,
\end{align*}
where $I = \mathrm{supp}(\ab) := \{ i \in [n] : a_i \not= 0 \}$ and
$[n] := \{1,\ldots,n\}$.  Note that $\mathcal{S}(\ab,b)$ has exactly
$2^{|I|-1}$ inequalities, and that it is satisfied in $\Qbb$ by a
$\{0,1\}$-assignment to the $\xb$-variables if and only if
$\ab^\transpose \xb = b$ is satisfied in $\Fbb_2$ by the same
assignment.  For a system of $m$ linear equations $\Ab \xb = \bb$ as
above, let $\mathcal{S}(\Ab, \bb) := \bigcup_{i=1}^m
\mathcal{S}(\ab_i,b_i)$ as $\ab_i$ ranges over the rows of $\Ab$ and
$b_i$ ranges over the components of $\bb$. Note that this system has
at most $m2^w$ inequalities, where $w$ is the maximum number of
non-zero components in the rows of $\Ab$.

For an equation of the form $\ab^\transpose \xb = b$ as above, an
alternative way of writing the system $\mathcal{S}(\ab,b)$ is by
imposing the system of polynomial equalities
\begin{align*}
  & \prod_{i \in T} x_i \cdot \prod_{i \in I\setminus T} (1-x_i) = 0 &
  \text{ for all } T \subseteq I \text{ such that } |T| \equiv 1-b\!\!\mod 2.
\end{align*}
In the following, for $I \subseteq [n]$ and $T \subseteq I$, let
$M^I_T(\xb) := \prod_{i\in T} x_i \prod_{i \in I\setminus T} (1-x_i)$.
Such polynomials are called \emph{extended monomials}.
We start by noting that
\begin{equation}
  \sum_{T \subseteq I} M^I_T(\xb) = \prod_{i \in I}(x_i + 1 - x_i) = 1.
\label{eq:fullsum}
\end{equation}

We continue relating the two forms of expressing the equation
$\ab^\transpose \xb = b$:

\begin{lemma}[Grigoriev, Hirsch, and Pasechnik
  \cite{GrigorievHirschPasechnik2002}] \label{lem:GHP2} Let $\ab \in
  \{0,1\}^n$ and $b \in \{0,1\}$, and let $I =
  \mathrm{supp}(\ab)$. For every $T \subseteq I$ such that $|T| \equiv
  1-b\!\!\mod 2$, the equation $M^I_T(\xb) = 0$ has a tree-like proof
  from $S(\ab,b)$ of length linear in $|I|$, and degree at most $|I|$.
\end{lemma}

\begin{proof} Let $t = |T|$ and assume without loss of generality that
  $T = \{1,\ldots,t\}$ and $I\setminus T = \{t+1,\ldots,s\}$. First
  multiply $\sum_{i=1}^t (1-x_i) + \sum_{i=t+1}^s x_i \geq 1$ by
  $x_1$. Then use $x_1 - x_1^2 = 0$ to get rid of the term $(1-x_1)
  \cdot x_1$ on the left-hand side. Repeat for $x_2,\ldots,x_t$ to get
  $\sum_{i = t+1}^s x_i\cdot \prod_{j = 1}^t x_j \geq \prod_{j = 1}^t
  x_j$.  From here, first multiply by $1-x_{t+1}$ and then use
  $x_{t+1} - x_{t+1}^2 = 0$ to get rid of $(1-x_{t+1})\cdot x_{t+1}
  \cdot \prod_{j=1}^t(1-x_j)$ on the left-hand side.  Repeat for
  $x_{t+2},\ldots,x_{s}$ to get $0 \geq \prod_{j = 1}^t x_i \cdot
  \prod_{j = t+1}^s (1-x_j)$.  The converse inequality has a direct
  proof not even using any of the axioms in $\mathcal{S}(\ab,b)$.
\end{proof}

The last lemma we need also refers to extended monomials:

\begin{lemma} \label{lem:thing}
  Let $T \subseteq I \subseteq [n]$. Then the equation 
$\big({\sum_{i \in I} x_i - |T|}\big) \cdot M^I_T = 0$ 
has a tree-like proof of length linear in $|I|$, and degree at most
$|I|+1$.
\end{lemma}

\begin{proof} Write $M$ for $M^I_T$. For every $i \in I\setminus T$,
  using $x_i \cdot (1-x_i) = x_i - x_i^2 = 0$ we get $x_i \cdot M =
  0$.  For every $i \in T$, using $x_i^2 - x_i = 0$ we get $x_i \cdot
  M = M$.  Adding up we get $\sum_{i \in I} x_i \cdot M = |T|\cdot M$.
\end{proof}

\begin{theorem} \label{thm:maintwoelementfield}
  Let $\Ab \in \{0,1\}^{m \times n}$ and $\bb \in \{0,1\}^m$.  If $\Ab
  \xb = \bb$ is unsatisfiable in $\Fbb_2$, then $S(\Ab,\bb)$ has a
  (not necessarily tree-like) refutation of size polynomial in $n$ and
  $2^w$, and degree linear in $w$, where $w$ is the maximum number of
  non-zero components in any of the rows of $\Ab$.
\end{theorem}

\begin{proof}
  Let $\ab_1,\ldots,\ab_m$ be the rows of $\Ab$, with $\ab_j =
  (a_{j,1},\ldots,a_{j,n})$.  Assume $\Ab \xb = \bb$ is unsatisfiable
  in $\Fbb_2$.  Then the $\Fbb_2$-rank of the matrix $[\Ab \;|\; \bb]$
  is bigger than the rank of $\Ab$. This means that there exists a
  subset of rows $J$ such that $|J| \leq n$ and $\sum_{j \in J} \ab_j
  = \zerob$ and $\sum_{j \in J} b_j = 1$ with arithmetic in
  $\Fbb_2$. In order to simplify notation, we assume without loss of
  generality that $J = \{1,\ldots,|J|\}$.

  For every $k \in \{0,\ldots,|J|\}$, let
  $$
  L_k(\xb) := \frac{1}{2} \left({ \sum_{j=1}^k \sum_{i=1}^n a_{j,i} x_i
  + \sum_{j=k+1}^{|J|} b_j}\right).
  $$
  We provide proofs of $D_c(L_k(\xb)) \geq 0$ for every $c \in R_k :=
  \{0,\ldots,(k+1)\cdot n \}$ by reverse induction on
  $k \in \{0,\ldots,|J|\}$.

  The base case $k = |J|$ is a special case of Lemma~\ref{lem:GHP}. To
  see why note that the condition $\sum_{j \in J} \ab_j = \zerob$
  means that if arithmetic is in $\Qbb$ then $\sum_{j \in J} a_{j,i}$
  is an even natural number for every $i \in [n]$. But then all the
  coefficients of
  $$
  L_{|J|}(\xb) = \frac{1}{2} \sum_{j=1}^{|J|} \sum_{i=1}^n a_{j,i} x_i
  = \sum_{i=1}^n \left({\frac{1}{2} \sum_{j=1}^{|J|} a_{j,i}}\right)
  x_i
  $$ 
  are integers. Hence Lemma~\ref{lem:GHP} applies.

  Suppose now that $0 \leq k \leq |J|-1$ and that we have a proof of
  $D_d(L_{k+1}(\xb)) \geq 0$ available for every $d \in R_{k+1}$.  Fix
  $c \in R_k$; our immediate goal is to give a proof of $D_c(L_k(\xb))
  \geq 0$.  As $k$ is fixed, write $L(\xb)$ instead of $L_{k+1}(\xb)$,
  and let the equation $\ab_{k+1}^\transpose \xb = b_{k+1}$ be written
  as $\sum_{i \in I} x_i = b$, where $I = \mathrm{supp}(\ab)$. Note
  that $L(\xb) = L_k(\xb) + \frac{1}{2}\cdot \ell(\xb)$ where
  $\ell(\xb) := -b + \sum_{i\in I} x_i$. Fix $T \subseteq I$ such that
  $|T| \equiv b\!\!\mod 2$, and let $d = c +\frac{t-b}{2}$ where $t =
  |T|$. Note that $d \in R_{k+1}$ as $c \in R_k$ and $0 \leq t \leq n$
  and $0 \leq b \leq 1$ are such that $t - b$ is even. Multiplying
  $D_{d}(L(\xb)) \geq 0$ by $M^I_T(\xb)$ we get
  \begin{equation}
    (L(\xb) - d) \cdot (L(\xb) - d + 1) \cdot M^I_T(\xb) \geq 0.
  \end{equation}
  Replacing $L(\xb) = L_k(\xb) + \frac{1}{2}\cdot \ell(\xb)$ in the factor
  $(L(\xb) - d)$ and recalling $d = c +\frac{t-b}{2}$, this
  inequality can be written as
  \begin{equation}
    (L_k(\xb) - c) \cdot (L(\xb) - d + 1) \cdot M^I_T(\xb)
    + (L(\xb) - d + 1) \cdot \textstyle{\frac{1}{2}} \cdot A(\xb) \geq 0
  \label{eq:inter1}
  \end{equation}
  where $A(\xb) := (\ell(\xb) + b - t) \cdot M^I_T(\xb)$. By
  Lemma~\ref{lem:thing} we have a proof of $A(\xb) = 0$, and hence of
  $(L(\xb)-d+1) \cdot \frac{1}{2}\cdot A(\xb) = 0$. Composing
  with~\eqref{eq:inter1} we get a proof of
  \begin{equation}
    (L_k(\xb) - c) \cdot (L(\xb) - d + 1) \cdot M^I_T(\xb) \geq 0.
  \end{equation}
  The same argument applied to the factor $(L(\xb)  - d + 1)$ of
  this inequality gives
  \begin{equation}
  (L_k(\xb) - c) \cdot (L_k(\xb) - c + 1) \cdot M^I_T(\xb) \geq 0.
  \end{equation}
  This is precisely $D_c(L_k(\xb)) \cdot M^I_T(\xb) \geq 0$. Adding
  up over all $T \subseteq I$ with $|T| \equiv b\!\!\mod 2$
  we get
  \begin{equation}
  D_c(L_k(\xb)) \cdot \sum_{T \subseteq I \atop |T| \equiv b}
  M^I_T(\xb) \geq 0. \label{eq:last1}
  \end{equation}
  By Lemma~\ref{lem:GHP2}, from the inequalities for $\sum_{i \in I}
  x_i = b$ we get proofs of $M^I_{T}(\xb) = 0$ for every $T \subseteq
  I$ such that $|T| \equiv 1-b\!\!\mod 2$.  But then also of
  $D_c(L_k(\xb)) \cdot M^I_T(\xb) = 0$ for every such $T$. Adding up and
  composing with~\eqref{eq:last1} we get
  $$
  D_c(L_k(\xb)) \cdot \sum_{T \subseteq I} M^I_T(\xb) \geq 0
  $$
  which is precisely $D_c(L_k(\xb)) \geq 0$ because $\sum_{T \subseteq
    I} M^I_T(\xb) = 1$ by~\eqref{eq:fullsum}.

  At this point we proved $D_c(L_0(\xb)) \geq 0$ for every $c \in R_0
  = \{0,\ldots, n \}$. Recall now that $\sum_{j=1}^{|J|} b_j$ is odd,
  say $2q+1$, and at most $n$. In particular $q+1$ belongs to $R_0$
  and $L_0(\xb) = q+\textstyle{\frac{1}{2}}$. Thus we have a proof of
  $D_{q+1}(L_0(\xb)) \geq 0$ where $D_{q+1}(L_0(\xb)) =
  -\textstyle{\frac{1}{2}}\cdot \textstyle{\frac{1}{2}} =
  -\textstyle{\frac{1}{4}}$.  Multiplying by $4$ we get $-1 \geq 0$.
\end{proof}
  
\section{Prime fields}

Let $p$ be a prime. We identify the elements of the field with $p$
elements $\Fbb_p$ with the integers $\{0,\ldots,p-1\}$. For every $i
\in [n]$, let $\xb_i = (x_i(0),\ldots,x_i(p-1))$ be formal variables
ranging over~$\Qbb$, and let $\xb = (\xb_1,\ldots,\xb_n)$. By imposing
the constraints
\[
\begin{array}{ll}
  x_i(0) + \cdots + x_i(p-1) = 1 & \text{for all } i \in [n] \\
  x_i(j) - x_i(j)^2 = 0 & \text{for all } i \in [n] \text{ and } j \in \{0,\ldots,p-1\}
\end{array}
\] 
each $\xb_i$ is the indicator vector of some value in
$\Fbb_p$. Consequently we think of $\xb_i$ as a formal variable
ranging over $\Fbb_p$.  In the following, let $\mathcal{Z}$ be the set
of equations $x_i(0) + \cdots + x_i(p-1) = 1$ as $i$ ranges over
$[n]$.

For every linear equation of the type $\ab^\transpose \xb = b$ where
$\ab = (a_1,\ldots,a_n) \in \Fbb_p^n$ and $b \in \Fbb_p$, let
$\mathcal{S}(\ab,b)$ be the system of linear inequalities
\[
\begin{array}{ll}
  \sum_{i \in I} (1-x_i(z_i)) \geq 1 & \text{ for all } \zb \in \Fbb_p^I \text{ such that } \sum_{i\in I} a_i z_i \not\equiv b\!\!\mod p
\end{array}
\] 
where $I = \mathrm{supp}(\ab) := \{ i \in [n] : a_i \not= 0 \}$.
Observe that these are at most $p^{|I|}$ different inequalities.  For
a system $\Ab \xb = \bb$ of $m$ linear equations as above, let
$\mathcal{S}(\Ab,\bb) := \bigcup_{i=1}^m \mathcal{S}(\ab_i,b_i)$ as
$\ab_i$ ranges over the rows of $\Ab$ and $b_i$ ranges over the
components of $\bb$.

In the following, for $I \subseteq [n]$ and $\zb \in \Fbb_p^I$,
let
\[
M_\zb(\xb) := \prod_{i \in I} \left({ x_i(z_i) \cdot \prod_{\ell=0: \atop
  \ell\not=z_i}^{p-1} (1-x_i(\ell))}\right).
\]

We start with the analogue of~\eqref{eq:fullsum}. This time we need
to assume some axioms.

\begin{lemma} \label{lem:complete} Let $I \subseteq [n]$. The equation
  $\sum_{\zb \in \Fbb_p^I} M_\zb(\xb) = 1$ has a tree-like proof from
  $\mathcal{Z}$ of length polynomial in $|I|$ and $p$, and degree linear
  in $|I|p$.
\end{lemma}

\begin{proof}
  Using $\sum_{\ell = 0}^{p-1} x_i(\ell) = 1$ and $x_i(z_i) -
  x_i(z_i)^2 = 0$ we have
\[
1 = \prod_{i\in I} \sum_{\ell=0}^{p-1} x_i(\ell) = \sum_{\zb \in
  \Fbb_p^k} \prod_{i\in I} x_i(z_i) = \sum_{\zb \in \Fbb_p^k} \prod_{i
  \in I} x_i(z_i)^{p} = \sum_{\zb \in \Fbb_p^k} \prod_{i \in I}
\left({x_i(z_i)\cdot \prod_{\ell=0: \atop \ell\not=z_i}^{p-1}
    (1-x_i(\ell))}\right).
\]
Use Lemma~\ref{lem:subs} to get an actual proof.
\end{proof}

\begin{lemma} \label{lem:fromaxioms} Let $\ab \in \Fbb_p^n$ and $b \in
  \Fbb_p$.  For every $\zb \in \Fbb_p^I$ such that $\sum_{i \in I}
  a_iz_i \not\equiv b\!\!\mod p$, where $I = \mathrm{supp}(\ab)$, the
  equation $M_\zb(\xb) = 0$ has a tree-like proof from
  $\mathcal{S}(\ab,b) \cup \mathcal{Z}$ of length polynomial in $|I|$
  and $p$, and degree at most $|I|p$.
\end{lemma}

\begin{proof}
  Without loss of generality, assume $I = \{1,\ldots,k\}$.  Start at
  $\sum_{i=1}^k (1-x_i(z_i)) \geq 1$ from $\mathcal{S}(\ab,b)$,
  multiply by $x_1(z_1)$, and use $x_1(z_1)\cdot (1-x_1(z_1)) =
  x_1(z_1) - x_1(z_1)^2 = 0$ to get $\sum_{i=2}^k (1-x_i(z_i)) \geq
  x_1(z_1)$. Repeat with $x_2(z_2),\ldots,x_k(z_k)$ to get $0 \geq
  \prod_{i=1}^k x_i(z_i)$. Multiply by $(1-x_i(\ell))$ for every $i
  \in I$ and $\ell \in \{0,\ldots,p-1\}\setminus\{z_i\}$ to get $0
  \geq M_\zb(\xb)$. The reverse inequality has a direct proof not even
  using any of the axioms from $\mathcal{S}(\ab,b) \cup \mathcal{Z}$.
\end{proof}

The following is the analogue of Lemma~\ref{lem:thing}:

\begin{lemma} \label{lem:identity}
Let $I \subseteq [n]$, $\ab \in \Fbb_p^I$, and $\zb \in \Fbb_p^I$.
Then the equation
\[
\left({\sum_{i\in I} a_i \sum_{\ell=0}^{p-1} \ell x_i(\ell) -
    \sum_{i \in I} a_i z_i}\right) \cdot M_\zb(\xb) = 0
\]
has a tree-like proof of length polynomial in $|I|$ and $p$, and
degree at most $|I|p+1$.
\end{lemma}

\begin{proof}
  Write $M$ for $M_\zb(\xb)$. For every $i \in I$ and every
  $\ell\in\{0,\ldots,p-1\}\setminus\{z_i\}$, using $x_i(\ell) \cdot
  (1-x_i(\ell)) = x_i(\ell) - x_i(\ell)^2 = 0$ we get $a_i \ell
  x_i(\ell) \cdot M = 0$. For every $i \in I$, using $x_i(z_i)^2 -
  x_i(z_i) = 0$ we get $a_i z_i x_i(z_i) \cdot M = a_i z_i \cdot
  M$. Adding up we get what we want.
\end{proof}

\begin{theorem} \label{thm:mainprimefield}
  Let $\Ab \in \Fbb_p^{m \times n}$ and $\bb \in \Fbb_p^m$.  If $\Ab
  \xb = \bb$ is unsatisfiable in $\Fbb_p$, then $\mathcal{S}(\Ab,\bb)
  \cup \mathcal{Z}$ has a (not necessarily tree-like) refutation of
  size polynomial in $n$ and $p^w$, and degree linear in $w$, where
  $w$ is the maximum number of non-zero components of the rows of
  $\Ab$.
\end{theorem}

\begin{proof}
  Let $\ab_1,\ldots,\ab_m \in \Fbb_p^n$ be the rows of $\Ab$. Assume
  $\Ab \xb = \bb$ is unsatisfiable in $\Fbb_p$.  Then the $\Fbb_p$-rank
  of the matrix $[\Ab \;|\; \bb]$ is bigger than the rank of $\Ab$,
  which means that there exists a subset of rows $J \subseteq [m]$ and
  a vector of multipliers $\yb = (y_j : j \in J) \in \Fbb_p^J$ such
  that $|J| \leq n$ and $\sum_{j \in J} y_j \ab_j = \zerob$ and
  $\sum_{j \in J} y_j b_j = 1$ with arithmetic in $\Fbb_p$. In order
  to simplify notation, we assume without loss of generality that $J =
  \{1,\ldots,|J|\}$.

  For every $k \in \{0,\ldots,|J|\}$, let
  $$
  L_k(\xb) := \frac{1}{p} \left({ \sum_{j=1}^k y_j\sum_{i=1}^n a_{j,i} X_i + \sum_{j=k+1}^{|J|} y_j b_j}\right),
  $$
  where $X_i := \sum_{\ell = 0}^{p-1} \ell \cdot x_i(\ell)$.  We
  provide proofs of $D_c(L_k(\xb)) \geq 0$ for every $c \in R_k :=
  \{0,\ldots,(k+1)p^2 n \}$ by reverse induction on $k \in
  \{0,\ldots,|J|\}$.

  The base case $k = |J|$ is a special case of Lemma~\ref{lem:GHP}. To
  see why note that the condition $\sum_{j \in J} y_j \ab_j = \zerob$
  means that if arithmetic is in $\Qbb$ then $\sum_{j \in J} y_j
  a_{j,i}$ is an integer multiple of $p$ for every $i \in [n]$. But then all
  the coefficients of
  $$
  L_{|J|}(\xb) = \frac{1}{p} \sum_{j=1}^{|J|} y_j \sum_{i=1}^n a_{j,i}
  X_i = \sum_{i=1}^n \sum_{\ell=0}^{p-1} \left({\frac{1}{p}
      \sum_{j=1}^{|J|} y_j a_{j,i}}\right) x_i(\ell)
  $$ 
  are integers. Hence Lemma~\ref{lem:GHP} applies.

  Suppose now that $0 \leq k \leq |J|-1$ and that we have a proof of
  $D_d(L_{k+1}(\xb)) \geq 0$ available for every $d \in R_{k+1}$.  Fix
  $c \in R_k$; our immediate goal is to give a proof of $D_c(L_k(\xb))
  \geq 0$.  As $k$ is fixed, write $L(\xb)$ instead of $L_{k+1}(\xb)$,
  and also $y$ instead of $y_{k+1}$,
  and let the equation $\ab_{k+1}^\transpose \xb = b_{k+1}$ be written
  as $\sum_{i \in I} a_i \xb_i = b$, where $I =
  \mathrm{supp}(\ab)$. Note that $L(\xb) = L_k(\xb) +
  \frac{y}{p} \cdot \ell(\xb)$ where $\ell(\xb) := -b + \sum_{i\in I} a_i
  X_i$.   

  Split $\Fbb_p^I$ into $Z := \{\zb \in \Fbb_p^I : \sum_{i\in I} a_i
  z_i \equiv b\!\!\mod p\}$ and $\overline{Z} := \Fbb_p^I \setminus
  Z$. Fix $\zb \in Z$ and let $t := \sum_{i\in I} a_i z_i$ with
  arithmetic in $\Qbb$. Let $d = c +\frac{(t-b)y}{p}$ and note that $d
  \in R_{k+1}$ as $c \in R_k$, $0 \leq t \leq p^2 n$, $0 \leq y \leq
  p-1$, and $0\leq b \leq p-1$ are such that $t-b$ is an integer
  multiple of $p$. Multiplying $D_{d}(L(\xb)) \geq 0$ by $M_\zb(\xb)$
  we get
  \begin{equation}
    (L(\xb) - d) \cdot (L(\xb) - d + 1) \cdot M_\zb(\xb) \geq 0.
  \end{equation}
  Replacing $L(\xb) = L_k(\xb) + \frac{y}{p} \cdot \ell(\xb)$ in the
  factor $(L(\xb) - d)$ and recalling $d = c +\frac{(t-b)y}{p}$, this
  inequality can be written as
  \begin{equation}
    (L_k(\xb) - c) \cdot (L(\xb) - d + 1) \cdot M_\zb(\xb)
    + (L(\xb) - d + 1) \cdot \textstyle{\frac{y}{p}} \cdot A(\xb) \geq 0
  \label{eq:inter}
  \end{equation}
  where $A(\xb) := (\ell(\xb) + b - t) \cdot M_\zb(\xb)$. By
  Lemma~\ref{lem:identity} we have a proof of $A(\xb) = 0$, and hence
  of $(L(\xb)-d+1) \cdot \frac{y}{p} \cdot A(\xb) = 0$. Composing
  with~\eqref{eq:inter} we get a proof of
  \begin{equation}
    (L_k(\xb) - c) \cdot (L(\xb) - d + 1) \cdot M_\zb(\xb) \geq 0.
  \end{equation}
  The same argument applied to the factor $(L(\xb)  - d + 1)$ of
  this inequality gives
  \begin{equation}
  (L_k(\xb) - c) \cdot (L_k(\xb) - c + 1) \cdot M_\zb(\xb) \geq 0.
  \end{equation}
  This is precisely $D_c(L_k(\xb)) \cdot M_\zb(\xb) \geq 0$. Adding
  over $Z$ we get
  \begin{equation}
  D_c(L_k(\xb)) \cdot \sum_{\zb \in Z}
  M_\zb(\xb) \geq 0. \label{eq:last}
  \end{equation}
  By Lemma~\ref{lem:fromaxioms}, from the inequalities in
  $\mathcal{S}(\ab_{k+1},b_{k+1})$ we get proofs of $M_\zb(\xb) = 0$
  for every $\zb \in \overline{Z}$. But then also $D_c(L_k(\xb)) \cdot
  M_\zb(\xb) = 0$ for every such $\zb$. Adding up and composing with
  \eqref{eq:last} we get
  \begin{equation*}
  D_c(L_k(\xb)) \cdot \sum_{\zb \in Z \cup \overline{Z}} M_\zb(\xb) \geq 0
  \end{equation*}
  which is precisely $D_c(L_k(\xb)) \geq 0$ because
  $\sum_{\zb\in\Fbb_p^I} M_\zb(\xb) = 1$ by Lemma~\ref{lem:complete}.

  At this point we proved $D_c(L_0(\xb)) \geq 0$ for every $c \in R_0
  = \{0,\ldots,p^2 n\}$. Recall now that $\sum_{j=1}^{|J|} y_j b_j$ is
  congruent to $1$ mod $p$, say $pq+1$, and smaller than $p^2 n$. In
  particular $q+1$ belongs to $R_0$ and $L_0(\xb) =
  q+\textstyle{\frac{1}{p}}$. Thus we have a proof of
  $D_{q+1}(L_0(\xb)) \geq 0$ where $D_{q+1}(L_0(\xb)) =
  (\textstyle{\frac{1}{p}}-1)\cdot \textstyle{\frac{1}{p}} =
  \textstyle{\frac{1-p}{p^2}}$.  Multiplying by
  $\textstyle{\frac{p^2}{p-1}} > 0$ we get $-1 \geq 0$.
\end{proof}

\section{Closing remarks}

The upper bound in Theorem~\ref{thm:maintwoelementfield} is to be put
in contrast with the lower bounds proved by Grigoriev
\cite{Grigoriev2001} as rediscovered by Schoenebeck
\cite{Schoenebeck2008}. Those lower bounds hold for \emph{static}
semi-algebraic proofs, and even static \emph{sums-of-squares} (SOS)
proofs.  In short, the static version of semi-algebraic proofs can be
formulated as the restriction to proofs in which all applications of
the multiplication rules must precede all applications of the positive
linear combination rule. Static sums-of-squares proofs would be the
same with the addition of axioms of the form $\sum_{i=1}^m P_i(\xb)^2
\geq 0$ for arbitrary polynomials $P_1,\ldots,P_m$. See
\cite{BarakBrandaoHarrowKelnerSteurerZhou2012} and subsequent work for
some recent exciting applications of static sums-of-squares proofs to
combinatorial optimization.

The above-mentioned lower bounds show that there exist explicit
systems of linear equations $\Ab \xb = \bb$ with $n$ variables and
three variables per equation, that are unsatisfiable over the
two-element field but for which any static semi-algebraic or
sums-of-squares refutation must have degree $\Omega(n)$. This holds
with respect to the same representation of linear systems that we use
here. It can also be seen that their proof also yields an exponential
$2^{\Omega(n)}$ lower bound in size and length.  More strongly, from
the size-degree trade-off results in \cite{PitassiSegerlind2012} for
tree-like proofs, such lower bounds on degree and size apply also to
the tree-like restrictions of semi-algebraic proofs and
sums-of-squares proofs. We note that static proofs may be assumed
tree-like without any significant loss in degree, size or length, so
this is a strengthening. Theorem~\ref{thm:maintwoelementfield} shows
that such lower bounds do not extend to general, i.e.~dag-like,
semi-algebraic proofs.

\end{document}